\documentclass[11pt,twoside]{article} 

\RequirePackage[OT1]{fontenc}
\RequirePackage{amsthm,amsmath,amsfonts,amssymb}
\RequirePackage[colorlinks]{hyperref}
\RequirePackage{epsfig, graphicx, color}
\RequirePackage{times}
\RequirePackage{latexsym}
\RequirePackage{psfrag}
\RequirePackage{color}
\usepackage{natbib}
\usepackage{xcolor}
\usepackage{cleveref}

\usepackage{fullpage}
\usepackage{epsf,epsfig}
\usepackage{fancyheadings}
\usepackage{graphics,graphicx}
\usepackage{enumerate}
\usepackage{algorithmicx,algorithm}
\usepackage{bbm}
\usepackage{mathtools}
\usepackage{tikz,subcaption,amsthm}

\usetikzlibrary{positioning,arrows,decorations,decorations.markings,decorations.pathreplacing,decorations.pathmorphing,shadows,positioning,arrows.meta,matrix,fit}
\usetikzlibrary{shapes,snakes}

\newtheorem{theorem}{Theorem}

\theoremstyle{definition}

\theoremstyle{remark}

\newcommand{\R}{\mathbb{R}}

\newcommand{\one}{\mathbbm{1}}

\def\cX{\mathcal{X}}

\def\pr{\mathbb{P}}

\def\cD{\mathcal{D}}

\newcommand{\vertiii}[1]{{\left\vert\kern-0.25ex\left\vert\kern-0.25ex\left\vert #1 
    \right\vert\kern-0.25ex\right\vert\kern-0.25ex\right\vert}}
    
\title{Localized conformal model selection}

\author{Yuhao Wang$^{1,2}$ and Tengyao Wang$^3$}
\date{$^1$ Tsinghua University\\
$^2$Shanghai Qizhi Institute\\
$^3$ London School of Economics and Political Science}
\begin{document}

\maketitle

\begin{abstract}
We propose a localized conformal model selection framework that integrates local adaptivity with post-selection validity for distribution-free prediction. By performing model selection symmetrically across calibration points using upper and lower surrogate intervals, we construct a data-dependent safe index set that contains the oracle model and preserves exchangeability. The resulting ensemble procedure retains exact finite-sample marginal coverage while adapting to spatial heterogeneity and model complexity. Simulations demonstrate substantial reductions in interval length compared to the best fixed model, especially in heterogeneous and low-noise settings.
\end{abstract}

\section{Introduction}

Uncertainty quantification is a central problem in modern statistical learning. Among available approaches, conformal prediction has emerged as a particularly attractive framework \citep{vovk2005algorithmic, shafer2008tutorial, lei2018distribution, angelopoulos2023conformal}. It provides finite-sample, distribution-free prediction intervals under minimal assumptions. These guarantees hold regardless of model complexity, making conformal prediction applicable to high-dimensional regression, nonparametric smoothing, and modern machine learning methods such as random forests and neural networks \citep{romano2019conformalized, barber2021predictive}.

In this paper, we study conformal prediction for regression problems in which the goal is to construct a prediction interval for a future response at a new covariate value, based on previously observed data drawn from an unknown distribution. In many practical settings, rather than relying on a single predictive model, practitioners consider a collection of candidate regression estimators that differ in smoothness, complexity, or modelling assumptions. Selecting or combining these models in a data-driven way is often crucial for achieving good predictive performance.

However, standard conformal prediction methods are designed for a fixed predictive model and do not directly account for the additional uncertainty introduced by model selection. When the same data are used both to select a model and to construct prediction intervals, the exchangeability assumptions underpinning conformal validity are typically violated. As a result, naively applying conformal prediction after model selection often leads to systematic undercoverage, a phenomenon that has been widely observed in practice \citep{liang2024conformal, hegazy2025valid}.

Recent work has begun to address this challenge. \citet{yang2025selection} and \citet{liang2024conformal} develop frameworks for conformal inference after model selection that preserves finite-sample coverage guarantees without sample splitting by correcting the selection bias inherent in naive procedures. Related ideas also appear in the growing literature on selective conformal inference, which studies post-selection validity under various data-dependent selection rules \citep{jin2025confidence}. While these approaches successfully control selection-induced bias, they operate at a global or population level and do not allow the selected model to vary across the covariate space. This limitation can be restrictive in nonparametric settings where the underlying regression function exhibits spatially heterogeneous structure.

A complementary line of work focuses on local adaptivity. Localized Conformal Prediction, introduced by \citet{guan2023localized}, modifies classical conformal prediction by weighting calibration residuals according to their proximity to the test point, yielding prediction intervals that adapt to local noise levels and model fit and can substantially reduce interval length in regions with simpler structure. Recent extensions explore related localization ideas, where local weights or randomization are used to achieve improved local coverage properties beyond marginal guarantees \citep{hore2023conformal}. More broadly, a growing literature investigates conditional or approximately conditional coverage guarantees, highlighting both the possibilities and fundamental limitations of achieving coverage that adapts to covariates \citep{foygel2021limits, gibbs2025conformal}. However, these localized methods assume a fixed predictive model and do not directly address the additional uncertainty arising from model selection or aggregation. 

The goal of this paper is to bridge these two strands of work by developing a conformal prediction framework that is both locally adaptive and valid after model selection. A naive approach would be to select, at each test point, the model that yields the shortest localized conformal interval. However, such a strategy breaks the symmetry required for conformal validity, as the model selection event depends implicitly on the unobserved response at the test point. Our main contribution is a new localized conformal model selection procedure that combines LCP with a safe indexing technique to restore the required symmetry. The key insight is that model selection must be performed symmetrically across all data points, rather than only at the final test point, in order to preserve exchangeability. By explicitly accounting for the unobserved response through upper and lower surrogate conformal intervals, we construct a data-dependent safe index set that is guaranteed to contain the oracle model, defined as the model yielding the shortest expected conformal interval.

This construction leads to a principled ensemble conformal prediction method that adapts both to local structure in the covariates and to model uncertainty, while retaining exact finite-sample coverage guarantees. The resulting prediction intervals automatically favor more complex models in highly oscillatory regions and simpler models in smoother regions, as illustrated in our numerical experiments. To the best of our knowledge, this is the first conformal inference framework that simultaneously achieves local adaptivity and post-selection validity in a fully nonparametric setting.

\section{A new localized conformal model selection}

We consider a regression setting in which we observe a calibration dataset 
\[
\cD := \{(X_i, Y_i)\}_{i=1}^n,
\]
consisting of independent and identically distributed observations from an unknown joint distribution on $\cX \times \R$. Our goal is to construct a prediction interval for an unobserved response $Y_{n+1}$ associated with a new covariate value $X_{n+1}$. Throughout, we assume that the augmented sample $\{(X_i,Y_i)\}_{i=1}^{n+1}$ is exchangeable. We are given a finite collection of candidate regression estimators $f_1,\ldots,f_K$, obtained from an independent training sample and are treated as fixed functions.

Since multiple regression estimators $\{f_k\}_{k=1}^K$ are available, an important question is how to combine them in order to construct shorter prediction intervals. 
We address this by employing \emph{localized conformal prediction (LCP)}~\citep{guan2023localized} and by introducing a safe-indexing refinement for model selection. Unlike the classical conformal prediction which constructs prediction sets via using the empirical quantile of regression residuals, the main idea of LCP is to produce a prediction interval for $X_{n+1}$ based on the estimated regression residuals $Y_i - f(X_i)$ and a localizer $H(\cdot, \cdot)$. Here localizer is used to up-weight the residuals whose $X_i$ are close to $X_{n + 1}$; in other words, LCP constructed prediction intervals based on a ``weighted'' empirical quantiles, where the weight of each sample is decided based on how close the sample covariance is to $X_{n + 1}$. As a result, the resulting prediction interval adapts to local noise levels and local model fit, rather than reflecting only global behaviour. Some typical localizers include the exponential kernel $H(X_{n + 1}, X_i) = \exp(-\|X_{n + 1} - X_i\|_2 / h)$ and the Gaussian kernel $H(X_{n + 1}, X_i) = \exp(-\|X_{n + 1} - X_i\|_2^2 / h^2)$ for suitable choice of localizer bandwidths $h$.

Suppose we are given $K$ conditional expectation estimates $f_1, \ldots, f_K$ (from a separate training sample) and a localizer $H$, then for each $f_k$, we may construct a conformalized prediction interval $C_k(\cdot)$ for $f_k$ based weighted empirical quantile of regression residuals, in a similar way as described in~\citet{guan2023localized}. More specifically, given a calibration set $\{(Y_i, X_i)\}_{i = 1}^n$ that is independent from the trained conditional expectation estimates $f_1,\ldots, f_K$, a test covariate $X$ and a model-specific miscoverage level $\gamma$, we may define $C_k(\{(Y_i, X_i)\}_{i = 1}^n, X, \gamma)$ as
\begin{equation}\label{eq:lci}
\begin{aligned}
C_k(\{(Y_i, X_i)\}_{i = 1}^n, & X, \gamma) := [f_k(X) - q, f_k(X) + q], \\
& \mathrm{where}\; q := \min\left\{v: \sum_{i = 1}^n \frac{H(X_i, X)}{\sum_{i = 1}^n H(X_i, X)} \one\{|Y_i - f_k(X_i)| \le v\} \ge 1 - \gamma \right\}.
\end{aligned}
\end{equation}

If the miscoverage level $\gamma$ is fixed in advance and a single model were used, this construction recovers standard localized conformal prediction. The challenge addressed in this paper is how to select or aggregate models without violating finite-sample conformal validity. For $i \in [n]$, define the leave-one-out dataset $\cD_{-i} := \cD \setminus \{(Y_i, X_i)\}$. Under the exchangeability assumption, it is natural to consider the conformalized prediction interval $C_k(\cD_{-i}\cup\{(Y_{n+1},X_{n+1})\}, X_i,\gamma)$ obtained by interchanging the roles of indices $i$ and $n+1$. Note that this is an ``oracle interval'' in the sense that depends on the unobserved response $Y_{n+1}$. In practice, we instead construct upper and lower surrogate intervals that correspond respectively to the hypothetical cases of extreme nonconformity and perfect conformity of $Y_{n+1}$:
\begin{align*}
    C_k^+(\cD_{-i}, X_i, \gamma) &:= C_k(\cD_{-i} \cup \{(\infty, X_{n + 1})\}, X_i, \gamma)\\
    C_k^-(\cD_{-i}, X_i, \gamma) &:= C_k(\cD_{-i} \cup \{(f_k(X_{n+1}), X_{n + 1})\}, X_i, \gamma).
\end{align*}

For each $i \in [n]$, we define the \emph{safe index set}
\[
\mathcal{K}_i(\gamma) := \left\{k: |C_k^-(\cD_{-i}, X_i, \gamma)| \le \min_{k'} |C_{k'}^+(\cD_{-i}, X_i, \gamma)|\right\}.
\]
This set contains all models whose best-case interval length is no larger than the worst-case interval length of the shortest competing model. By construction, 
\[
C_k^-(\cD_{-i},X_i,\gamma)
\subseteq
C_k\bigl(\cD_{-i}\cup\{(X_{n+1},Y_{n+1})\},X_i,\gamma\bigr)
\subseteq
C_k^+(\cD_{-i},X_i,\gamma),
\]
for every $k$.  Hence any model minimizing the oracle interval length must belong to $\mathcal K_i(\gamma)$.

Using the safe index sets $\mathcal{K}_i(\gamma)$, we determine lower and upper bounds on the admissible miscoverage level by evaluating empirical coverage across all calibration points. Specifically, over a discrete grid $\Gamma \subseteq [0,1]$, we define
\begin{equation}
\begin{aligned}
\hat{\gamma}_L(\alpha) & := \max \left\{\gamma \in \Gamma: \frac{1}{n + 1} \sum_{i = 1}^n \one\left\{Y_i \in \cap_{k \in \mathcal{K}_i(\gamma)} C^-_k(\cD_{-i}, X_i, \gamma)\right\} \ge 1 - \alpha\right\}, \\
\hat{\gamma}_U(\alpha) & := \max \left\{\gamma \in \Gamma: \frac{1}{n + 1} \sum_{i = 1}^n \one\left\{Y_i \in \cup_{k \in \mathcal{K}_i(\gamma)} C^+_k(\cD_{-i}, X_i, \gamma)\right\} \ge 1 - \alpha - \frac{1}{n + 1}\right\}.
\end{aligned}
\label{Eq:GammahatLU}
\end{equation}

Finally, the prediction interval for $Y_{n+1}$ is obtained by selecting, for each admissible miscoverage level, the shortest localized conformal interval and taking their union. Formally,
\begin{equation}\label{eq:finalcp}
C_{\mathrm{LCP-MS}}(\cD, X_{n + 1}, \alpha) := \bigcup_{\gamma \in [\hat{\gamma}_L(\alpha), \hat{\gamma}_U(\alpha)] \cap \Gamma} C_{\min}(\cD, X_{n + 1}, \gamma),
\end{equation}
where $C_{\min}(\cD, X_{n + 1}, \gamma)$ denotes the shortest interval among $\{C_k(\cD, X_{n + 1}, \gamma)\}_{k=1}^K$, with ties broken by index order. 

The following theorem establishes the finite-sample validity of the resulting procedure.

\begin{theorem}
Consider a dataset $\{(Y_i, X_i)\}_{i = 1}^{n + 1}$ that are i.i.d.\ realizations from a distribution $\pr$, where $Y_{n + 1}$ is unobserved. Given a sequence of fixed functions $f_1, \ldots, f_K$, we have that the $C_{\mathrm{LCP-MS}}(\cdot)$ constructed in~\eqref{eq:finalcp} satisfies
\[
\pr(Y_{n + 1} \in C_{\mathrm{LCP-MS}}(\cD, X_{n + 1}, \alpha)) \ge 1 - \alpha
\]
for any $\alpha \in [0, 1]$.
\end{theorem}

\begin{proof}
    Let $\tilde\cD:=\cD\cup\{(X_{n+1},Y_{n+1})\}$ and 
$\tilde\cD_{-i}:=\tilde\cD\setminus\{(X_i,Y_i)\}$. Define
    \begin{equation}\label{eq:hatalpha}
    \hat{\gamma} := \max\left\{\gamma \in \Gamma \;:\; \frac{1}{n + 1} \sum_{i=1}^{n + 1} \one\{Y_i \in C_{\min}(\tilde\cD_{-i}, X_i, \gamma)\} \ge 1 - \alpha\right\}.
    \end{equation}
    Set
\[
Z_i
:=
\one\{Y_i\in C_{\min}(\tilde\cD_{-i},X_i,\hat\gamma)\}.
\]
    By exchangeability of $\{(X_i,Y_i)\}_{i=1}^{n+1}$ and the permutation invariance of the construction, the variables $Z_1,\ldots,Z_{n+1}$ are exchangeable. By definition of $\hat\gamma$, and using exchangeability, we have 
\[\pr\bigl(
Y_{n+1}\in C_{\min}(\cD,X_{n+1},\hat\gamma)
\bigr) = 
\pr(Z_{n+1}=1)
=
\mathbb E\!\left[\frac{1}{n+1}\sum_{i=1}^{n+1}Z_i\right]
\ge1-\alpha.
\]
Thus, it remains to show that $\hat{\gamma}_L(\alpha) \le \hat{\gamma} \le \hat{\gamma}_U(\alpha)$ for $\hat\gamma_L(\alpha)$ and $\hat\gamma_U(\alpha)$ defined in~\eqref{Eq:GammahatLU}.

To this end, from the surrogate inclusions, we have for all $i \in [n]$ and $\gamma \in\Gamma$ that 
    \[
C_k^-(\cD_{-i}, X_i, \gamma) \subseteq 
C_k(\tilde{\cD}_{-i}, X_i, \gamma) \subseteq C_k^+(\cD_{-i}, X_i, \gamma).
\]
Hence the model $f_k$ whose corresponding interval $C_k(\tilde{\cD}_{-i}, X_i, \gamma)$ is shortest among all $C_{k'}(\tilde{\cD}_{-i}, X_i, \gamma)$ for $k\in\{1,\ldots,K\}$ must belong to the set $\mathcal{K}_i(\gamma)$. Consequently, for all $i \in [n]$, we have
\[
\cap_{k \in \mathcal{K}_i(\gamma)} C^-_k(\cD_{-i}, X_i, \gamma) \subseteq 
C_{\min}(\tilde{\cD}_{-i}, X_i, \gamma) \subseteq \cup_{k \in \mathcal{K}_i(\gamma)} C^+_k(\cD_{-i}, X_i, \gamma).
\]
Then it follows from above,~\eqref{eq:hatalpha}, as well as the definitions of $\hat{\gamma}_L(\alpha)$ and $\hat{\gamma}_U(\alpha)$ in~\eqref{Eq:GammahatLU} that $\hat{\gamma}_L(\alpha) \le \hat{\gamma} \le \hat{\gamma}_U(\alpha)$, which proves the desired result.
\end{proof}

\section{Numerical simulations}
We illustrate the proposed localized conformal model selection (LCP-MS) method through two simulation studies: a parametric setting with locally oscillatory signals and a nonparametric regression setting with varying curvature. In all experiments, we report the average conformal interval length at the target coverage level $1-\alpha=0.9$ based on 100 replications. We compare the proposed LCP-MS ensemble with the best-performing single model chosen in hindsight (denoted ``best\_single'').

We first consider a piecewise sinusoidal model defined as
\[
Y = 
\begin{cases}
\sin(5X) + \varepsilon, & X \in [-5, 0),\\
2\sin(3X) + \varepsilon, & X \in [0, 5],
\end{cases}
\qquad \varepsilon \sim N(0, \sigma^2),
\]
where the regression function changes its frequency and amplitude at $X=0$. 

We fit a family of local parametric models 
\[
M_{\lambda,h}: \quad Y = A\sin(\lambda X + \phi),
\]
where the amplitude $A$ and phase shift $\phi$ are estimated via least squares using only data within a local window $[x-h,\,x+h]$. Each model is indexed by a frequency parameter $\lambda \in \{1,2,3,4,5\}$ and a local bandwidth $h \in \{0.5,1\}$, resulting in 10 candidate models in total. 

We vary the noise level $\sigma \in \{0.1,0.3\}$ and the training/calibration sample size $n \in \{200,500,1000\}$. For each configuration, Table~\ref{tab:parametric} reports the average length of the conformal prediction intervals obtained by the proposed LCP-MS ensemble and by the best individual model.

\begin{table}[ht]
\centering
\caption{Average conformal interval lengths in the parametric setting.}
\label{tab:parametric}
\begin{tabular}{ccccc}
\hline\hline
$n$ & $\sigma$ & localizer bw & ensemble\_len & best\_single \\
\hline
200 & 0.1 & 0.1 & 0.3985 & 0.4865\\
500 & 0.1 & 0.1 & 0.3660 & 0.4445\\
1000 & 0.1 & 0.1 & 0.3470 & 0.4305\\
200 & 0.3 & 0.1 & 1.0920 & 1.1020\\
500 & 0.3 & 0.1 & 1.0950 & 1.0460\\
1000 & 0.3 & 0.1 & 1.0455 & 1.0330\\
200 & 0.1 & 0.3 & 0.3840 & 0.4935\\
500 & 0.1 & 0.3 & 0.3550 & 0.4460\\
1000 & 0.1 & 0.3 & 0.3450 & 0.4295\\
200 & 0.3 & 0.3 & 1.1620 & 1.1130\\
500 & 0.3 & 0.3 & 1.0475 & 1.0480\\
1000 & 0.3 & 0.3 & 1.0155 & 1.0385\\
\hline\hline
\end{tabular}
\end{table}

Across all sample sizes, the proposed LCP-MS ensemble consistently produces shorter intervals than the best single model under low noise ($\sigma=0.1$), demonstrating its benefit from adaptive local model selection. Under higher noise ($\sigma=0.3$), both methods perform similarly, indicating that the advantage of localization diminishes when the signal-to-noise ratio is low. 

We next consider a smooth nonparametric regression model
\[
Y = \sin(X^3) + \varepsilon, \qquad X \in [0,3], \quad \varepsilon \sim N(0, \sigma^2),
\]
where the curvature of the mean function varies with $X$. We fit five Nadaraya--Watson kernel regression estimators with Gaussian kernels and bandwidths $h \in \{0.1, 0.2, 0.4, 0.8, 1.6\}$. The same LCP-MS procedure is applied to combine these models.

We vary $\sigma \in \{0.1, 0.3\}$ and $n \in \{200, 500, 1000, 2000\}$. Table~\ref{tab:nonparametric} summarizes the average interval lengths. In this nonlinear setting, the LCP-MS intervals are substantially shorter, roughly 25–35\% reduction compared with the best single model, while maintaining nominal coverage. The improvement is more pronounced for small $\sigma$ and large $n$, highlighting the ability of the proposed method to adapt to local structure and select the most appropriate kernel bandwidths across regions of differing smoothness.

\begin{table}[htbp]
\centering
\caption{Average conformal interval lengths in the nonparametric setting.}
\label{tab:nonparametric}
\begin{tabular}{ccccc}
\hline\hline
$n$ & $\sigma$ & localizer bw & ensemble\_len & best\_single \\
\hline
200 & 0.1 & 0.3 & 0.9500 & 1.3665\\
500 & 0.1 & 0.3 & 0.9245 & 1.3265\\
1000 & 0.1 & 0.3 & 0.9045 & 1.3175\\
2000 & 0.1 & 0.3 & 0.9000 & 1.3225\\
200 & 0.3 & 0.3 & 1.4580 & 1.5790\\
500 & 0.3 & 0.3 & 1.4145 & 1.5495\\
1000 & 0.3 & 0.3 & 1.4170 & 1.5530\\
2000 & 0.3 & 0.3 & 1.3960 & 1.5390\\
\hline\hline
\end{tabular}
\end{table}

\begin{figure}[htbp]
    \centering
    \includegraphics[width=0.5\linewidth]{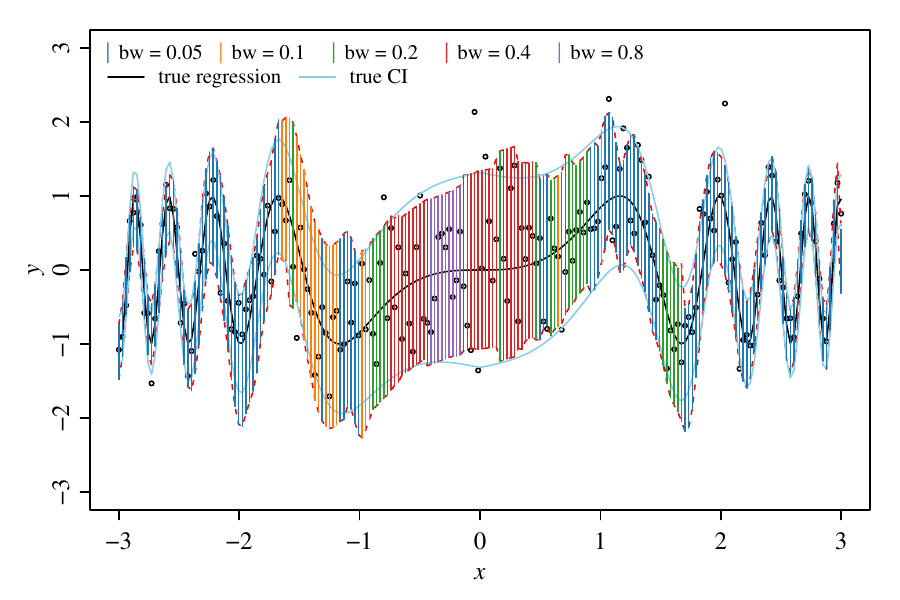}
    \caption{Localized ensemble conformal prediction intervals evaluated at each test point. The black curve denotes the true mean regression function, while the light blue curves indicate the true 90\% confidence bands. Vertical bars represent ensemble conformal prediction intervals, with colours corresponding to different models selected within the ensemble. In highly oscillatory regions, models with smaller bandwidths are predominantly selected (blue), whereas in smoother regions the ensemble favours models with larger bandwidths (red and purple).}
    \label{fig:nadaraya}
\end{figure}

To illustrate the adaptive behaviour underlying the improvements shown in the simulation results, we focus on a representative configuration with $n = 500$, $\sigma = 0.1$ and $h = 0.3$ and examine the resulting localized ensemble conformal intervals in Figure~\ref{fig:nadaraya}. The figure shows that the LCP-MS ensemble adaptively selects different kernel bandwidths across the covariate domain in accordance with the local geometry of the true regression function. In regions where the curvature is high, smaller bandwidth models are predominantly chosen, producing tighter intervals that closely track rapid changes in the signal. In contrast, in smoother regions the ensemble favours larger bandwidths, yielding more stable and efficient inference. This spatially varying model selection provides a clear explanation for the superior performance of the ensemble observed in Table~\ref{tab:nonparametric}, as no single fixed-bandwidth estimator can adapt uniformly well across regions of differing smoothness.

\section{Discussion}
The main technical challenge addressed in the paper is to reconcile local, data dependent model selection with the symmetry requirements of conformal inference. Selecting the shortest localized interval at a test point alone breaks exchangeability. The safe indexing construction resolves this by enforcing selection rules simultaneously across calibration points through surrogate intervals. 

We remark that while the procedure improves interval adaptivity across the covariate space, the overall guarantee remains marginal rather than conditional. Efficiency gains are most pronounced when the candidate family contains a model that approximates the local regression function well. In such settings, and when the best and second best models are well separated, the safe index set $\mathcal{K}_i(\gamma)$ is expected to become asymptotically a singleton, so the method effectively behaves like a locally optimal selector.

The procedure is computationally heavier than standard conformal prediction due to leave one out constructions and evaluation over multiple models and miscoverage levels, though these computations are naturally parallelizable. While developed for regression, the symmetry based safe indexing idea may extend to other adaptive conformal settings where data dependent decisions threaten exchangeability.

\bibliographystyle{custom}
\bibliography{reference}
\end{document}